\documentclass[pra,twocolumn,superscriptaddress,floatfix,amsmath,nofootinbib,amssymb]{revtex4-1}
\usepackage[UKenglish]{babel}
\usepackage{graphicx}% Include figure files
\usepackage{dcolumn}% Align table columns on decimal point
\usepackage{bm}% bold math
\usepackage{verbatim}
\usepackage{mathrsfs}
\usepackage{color}

\newcommand{\ket}[1]{\left| {#1} \right\rangle}
\newcommand{\bra}[1]{\left\langle {#1} \right|}

\newcommand{\proj}[2]{\left| {#1} \right\rangle\!\left\langle {#2} \right|}

\newcommand{\tr}{\operatorname{Tr}}

\newcommand{\eq}[1]{(\ref{#1})}

\newcommand{\up}{\uparrow}
\newcommand{\down}{\downarrow}

\newtheorem{theorem}{Theorem}[section]

\newenvironment{proof}[1][Proof]{\begin{trivlist}
\item[\hskip \labelsep {\bfseries #1}]}{\end{trivlist}}
\newenvironment{definition}[1][Definition]{\begin{trivlist}
\item[\hskip \labelsep {\bfseries #1}]}{\end{trivlist}}

\newcommand{\qed}{\nobreak \ifvmode \relax \else
      \ifdim\lastskip<1.5em \hskip-\lastskip
      \hskip1.5em plus0em minus0.5em \fi \nobreak
      \vrule height0.75em width0.5em depth0.25em\fi}

\begin{document}
\title{Comment on ``On two misconceptions in current relativistic quantum information''}
\author{Miguel Montero}
\affiliation{Instituto de F\'{i}sica Fundamental, CSIC, Serrano 113-B, 28006 Madrid, Spain}
\author{Eduardo Mart\'{i}n-Mart\'{i}nez}
\affiliation{Instituto de F\'{i}sica Fundamental, CSIC, Serrano 113-B, 28006 Madrid, Spain}

\begin{abstract}
We show that, contrarily to the claims in [K. Bradler, (2011), arXiv:1108.5553] (and arXiv:1201.1045), the reference [Phys. Rev. A 83, 062323 (2011)] is devoid of flaws or mistakes. We show that the criticism comes from a misunderstanding of part of our work. We also discuss the claims made in arXiv:1108.5553 about entanglement measures as a valid tool in relativistic quantum information. 
\end{abstract}

\maketitle

In \cite{Mig2} we showed that some previous results on fermionic entanglement in non-inertial frames were not physical due to the arbitrary endowing of a tensor product structure on the fermionic field before computing entanglement measures. We found that this scheme of endowing a tensor product structure to the Fock space may still yield physical results if a particular tensor product structure (different from that of the literature) was chosen.

In \cite{br}, K. Br\'adler claimed that our work is flawed. His criticisms seem to reflect that he understood our work as claiming that there is a real ambiguity in physical results  when studying entanglement in noninertial frames. This is not our claim, and in this work we intend to show explicitly where the misunderstandings lie. We will show that the formalism in \cite{Mig2} is exactly equivalent to the derivation provided in \cite{br}. K. Br\'adler and R. J\'auregui recently published a comment to \cite{Mig2} that appeared in \cite{br2}. The contents of \cite{br2} are fully contained in \cite{br}, and hence this paper constitutes a contestation to both works.

Indeed, in \cite{Mig2} we do not claim that there is an ambiguity in the way entanglement is defined in fermionic systems, but rather, it is in the way it had been computed by some previous works on field entanglement in noninertial frames. After considering \cite{br}, we feel that this crucial point was not clear enough in our original article.

It is true that seeking to illustrate the problem with simple examples of the phenomenon we wanted to show, we chose field states not complaining with the fermionic superselection rule. This was regarded by \cite{br} as a fatal flaw of our work. In this reply we show that it was only was an academic choice of the examples for the sake of simplicity. Indeed, the particular states used in these examples are irrelevant for the results of \cite{Mig2}, so we give here similar examples but now compliant with the rule showing that this fact has no impact on our claims.

In section \ref{TPS} we explicitly show the equivalence of the formalism in \cite{br} with the one used in \cite{Mig2}, both giving exactly the same results. In section \ref{super} we discuss the fermionic superselection rule as applied to our states, and reviews the negativity as a valid entanglement measure for fermionic systems, contrasting our views with those of \cite{br}. Finally, section \ref{cojoniak} presents our conclusions.

\section{Tensor product structures and entanglement}\label{TPS}

We will show that the final results of \cite{Mig2} (see sec. IV of that work) are consistent with the fermionic nature of the field just as well as the results in \cite{br}, contrary to the claims in \cite{br}. 
 
First, the author of \cite{br} correctly asserts that there is no natural tensor product structure in a multipartite fermionic system, as this would conflict with the canonical anticommutation relations. \cite{br} considers that these different results appeared only because we did not respect the canonical anticommutation relations.

However, this is not the case. The fact that there is no natural tensor product structure in a fermionic system is a well-known problem  and especially manifests in quantum information theory when partitions of fermionic systems have to be defined.

Not even the braided tensor product suggested in \cite{br}  is exempt of these issues when the problem of separability comes forth. 
As a concrete  example, consider the fermionic 4-mode state 
\begin{align}\ket{\Psi}&=\frac{1}{2}\left(\ket{1_a1_c}+\ket{1_a1_d}+\ket{1_b1_c}+\ket{1_b1_d}\right)\nonumber\\&=\left[\frac{1}{\sqrt{2}}\left(\ket{1_a}+\ket{1_b}\right)\right]\tilde\otimes\left[\frac{1}{\sqrt{2}}\left(\ket{1_c}+\ket{1_d}\right)\right].\label{st1}\end{align}
which factors with respect to the braided tensor product when expressed in the basis 
\begin{align}\ket{1_a1_c}&=a^\dagger \tilde\otimes\,\; c^\dagger \ket{0},\; \ket{1_b1_c}=b^\dagger \tilde\otimes\,\; c^\dagger\ket{0},\nonumber\\ \ket{1_a1_d}&=a^\dagger \tilde\otimes\,\; d^\dagger\ket{0},\ket{1_b1_d}=b^\dagger\tilde\otimes\,\; d^\dagger\ket{0}.\label{bt1}\end{align}
This state can be also expressed in the basis 
\begin{align}\ket{1_a1_c}'&=a^\dagger \tilde\otimes\,\; c^\dagger \ket{0},\; \ket{1_b1_c}'=b^\dagger \tilde\otimes\,\; c^\dagger\ket{0},\nonumber\\ \ket{1_a1_d}'&=a^\dagger \tilde\otimes\,\; d^\dagger\ket{0},\ket{1_d1_b}'=-d^\dagger\tilde\otimes\,\; b^\dagger\ket{0}.\label{bt2}\end{align}
that differs in one sign due to the anti-symmetry properties of fermionic systems (we have permuted the labels of first and second subsystem). Then
it reads
\begin{align}\ket{\Psi}&=\frac{1}{2}\left(\ket{1_a1_c}+\ket{1_a1_d}+\ket{1_b1_c}-\ket{1_d1_b}\right),\end{align}
which is not factorisable (in fact it is maximally entangled with respect to the basis defined by the braided product). Of course we are not extracting any conclusions about the physical meaning of this factorisability, only showing that problems with multipartite fermionic tensorial structures are present even if we use the formalism in  \cite{br} and even if the state has no superselection rule issues. In the first version of \cite{br}, a similar example was disregarded due to the lack of compliance with the fermionic superselection rule discussed in \cite{br} and below. Version 2 of \cite{br} addresses the example \eqref{st1} and correctly explains it by asserting that entanglement is partition-dependent. This is precisely our conclusion: different fermionic operator orderings may result in different bipartitions and thus different entanglement behaviours. %We are glad to see that \cite{br} now agrees on this point.

However, we are not claiming that the formalism used in \cite{br} is incorrect, on the contrary, we claim that it is equivalent to what we do in \cite{Mig2}. In that article we simply use a different approach than that in \cite{br} which is equally valid since all the physical conclusions are the same. To show that, we now proceed to describe our approach in detail from a somewhat more mathematical viewpoint than that used in \cite{Mig2}; to do so, we start with some definitions.

We shall consider a particular set of $m+n$ fermionic modes. The annihilation operators of these modes will be labeled as $a_i$ and $ c_j$, with $i=\{1\ldots n\}$, $j=\{1\ldots m\}$. These modes span a $2^{m+n}$-dimensional fermionic Fock space when acting on their vacuum state $\ket{0}_a\tilde\otimes\ket{0}_c$. We can define the usual particle number basis as follows:
\begin{align}&{\underbrace{\ket{1\ldots1}}_{m+n}}\ _{a_1\ldots a_n  c_1 \ldots  c_m}\nonumber\\&=\left(a_1\right)^\dagger\ldots\left(a_n\right)^\dagger\left( c_1 \right)^\dagger\dots\left( c_m\right)^\dagger\ket{0}_a\tilde\otimes\ket{0}_c\label{1}\end{align}
and for any basis elements with a $0$ instead of a $1$, we just remove the corresponding operator on the r.h.s.

We also consider a field state $\rho$ over this Hilbert space. This state will be uniquely determined by its matrix elements in the basis described above, so that
\begin{align}\rho=\sum_{k_a,l_a,k_c,l_c}\rho_{k_a,l_a,k_c,l_c}\proj{k_ak_c}{l_cl_a}\label{st}\end{align}
where $k_a$, $k_c$ are strings of 0's and 1's of length $n$ and $l_a$ and $l_c$ are strings of 0's and 1's of length $m$. Here, it is worth noticing that 
\begin{align}(\ket{l}_{a_1\ldots a_n  c_1 \ldots  c_m})^\dagger=\bra{l}_{ c_m\ldots  c_1 a_n \ldots a_1}.\label{xs}\end{align}
Note the reordering of the mode labels in \eq{xs}.%, which is a simple restatement of the property that $(AB)^\dagger=B^\dagger A^\dagger$.

Finally, we introduce a mapping from this fermionic Fock space to a qubit Hilbert space of $m+n$ modes:
\begin{definition} \label{qdef}Let $a^\text{Bos.}_i$ and $c^\text{Bos.}_j$ be two sets of bosonic operators acting on their vacuum state $\ket{0}$. For each fermionic state $\rho$, we define $\mathcal{Q}(\rho)$ as the state formally obtained as follows: First, write down $\rho$ in the basis \eq{1} (i.e. as in \eq{st}). Then, formally replace each $a_i$ with $a_i^\text{Bos.}$, $ c_j$ with $ c_j^\text{Bos.}$ and the vacuum state $\ket{0}$ with $\ket{0}$.
\end{definition}

We will denote the inverse mapping by $\mathcal{Q}^{-1}$. The restriction of this mapping to the subspace spanned by $a^\text{Bos.}$ modes on $\ket{0}$ will be denoted as $\mathcal{Q}^{-1}\vert_a$.

Note that this definition relies heavily on the use of the basis $\eq{1}$. We can define analogous mappings which are different to $\mathcal{Q}$ and which rely on different bases, obtained e.g. by permuting the operators in \eq{1}. This way we obtain a family of different mappings, each of them associated to an `operator ordering' . The mapping $\mathcal{Q}$ corresponds to what we call the `physical ordering' in \cite{Mig2}. 

The map $\mathcal{Q}$ can be said to be equivalent to the `endowing' of a particular tensor product structure on the fermionic Fock space. This formal endowing has been widely used in the literature  \cite{AlsingSchul,Edu2,Edu3,allicupahc,allicupahc2,Shapoor,Geneferm,chor1,Edu9,chor2,Cirac}. We are aware of the fact, pointed out in \cite{br}, that fermions and bosons correspond to different group representations; this is not an obstacle for the definition of $\mathcal{Q}$. Of course, if care is not taken in the choice of basis \eq{1}, this procedure leads to unphysical results for reduced states and entanglement measures among other things. \cite{Mig2} studies when and how this procedure can be used to obtain correct, physical results. These results are the same as those which would be obtained with the braided formalism of \cite{br}.

Finally, we remark that the key point in the definition of $\mathcal{Q}$ is the relative ordering between the $a$ and $c$ operators. A permutation of $c$-operators in \eq{1} is  irrelevant. If we permute, for instance, the positions of $c_3$ and $c_2$, then it is a simple matter of defining the basis as
\begin{align}&{\underbrace{\ket{1\ldots1}'}_{m+n}}\ _{a_1\ldots a_n  c_1 c_3 c_2 \ldots  c_m}=\nonumber\\&\left(a_1\right)^\dagger\ldots\left(a_n\right)^\dagger\left( c_1\right)^\dagger \left( c_3\right)^\dagger \left( c_2\right)^\dagger\dots\left( c_m\right)^\dagger\ket{0}_a\tilde\otimes\ket{0}_c\label{12},\end{align}
and define $\mathcal{Q}$ using this basis.

Now, for some reason, we are interested in tracing out the $c$ modes. In the context of \cite{Mig2}, this is because field modes living in region II of Rindler spacetime are causally disconnected from the observers.  This is implemented by the partial trace, which may be defined as
\begin{align}&\rho\rightarrow \tr_c(\rho)\equiv\nonumber\\&\sum_{j_1,j_2\ldots j_m\in\{0,1\}}\bra{0}_c( c_m)^{j_m}\ldots( c_1)^{j_1}\rho(c^\dagger_1)^{j_1}\ldots(c^\dagger_m)^{j_m}\ket{0}_c\label{2}\end{align}

Note that \eq{2} may give ambiguous results for states not complying with the superselection rule of \cite{br}, but it defines an unique reduced state for states within a given superselection sector.

%The partial trace \eq{2} can be obviously carried out in any basis  we want (corresponding for instance to other operator orderings). In most cases, there will appear some anticommutation signs. These signs are not accounted for if one  basis-dependent mapping: If these signs appear, then the mapping to a qubit space is deemed to give wrong results. This is the original motivation for our paper.

We now prove that the same result is obtained through direct application of \eq{2} or, equivalently, through the use of the $\mathcal{Q}$ mapping (i.e. the tensor product structure associated to the `physical ordering'), which is what we do in our paper:
\begin{theorem} Let $\rho$ be a fermionic state compliant with the fermionic superselection rule. Then the equality
\begin{equation*}\tr_c(\rho)=\mathcal{Q}^{-1}\vert_a\circ\tr^\text{Bos.}_c(\mathcal{Q})(\rho)\end{equation*}
holds, where $\tr^\text{Bos.}_c$ is the usual partial trace operation in a qubit system.\end{theorem}

\begin{proof}
Since by virtue of the definition \eq{2} the fermionic partial trace is a linear operator, we can consider each of the terms in \eq{st} independently. For each of these terms, we have
\begin{align}&\tr_c{\proj{k_ak_c}{l_cl_a}}=\nonumber\\&\delta_{k_c,l_c}( c_m)^{j_1}\ldots\bra{0}_c( c_1)^{j_1}(a^\dagger_1)^{i_1}\ldots(a^\dagger_n)^{i_n}(c^\dagger_1)^{j_1}\ldots(c^\dagger_m)^{j_m}\nonumber\\&\proj{0}{0}( c_m)^{j_m}\ldots( c_1)^{j_1}(a_n)^{k_n}\ldots(a_1)^{k_1}(c^\dagger_1)^{j_1}\ldots(c^\dagger_m)^{j_m}\ket{0}_c\nonumber\\&=(a^\dagger_1)^{i_1}\ldots(a^\dagger_n)^{i_n}\proj{0}{0}(a_n)^{k_n}\ldots(a_1)^{k_1}\nonumber\\&=\delta_{k_c,l_c}\proj{k_a}{l_a}\label{pro}\end{align}
the last equality holding because in the second line of \eq{pro} the number of $a$-operators to the left and  the number of $a$-operators to the right of $\proj{0}{0}$ are of the same parity, due to the superselection rule. Hence, both left and right operators give the same global sign when all the anticommutators are computed, giving no extra anticommutation signs.

The trace on $\mathcal{Q}(\rho)$ is also a linear operator, so we also consider the terms independently. In this case, no sign issues appear and the reduced state is simply $\proj{k_a}{l_a}_\text{Bos.}$, where the subscript reminds that this is still a qubit state. But now, $\mathcal{Q}^{-1}\vert_a$ maps $\proj{k_a}{l_a}_\text{Bos.}$ to $\proj{k_a}{l_a}$, which ends the proof.\qed\end{proof}

Note that this theorem implies that any operation being a function of the fermionic matrix elements only (for instance, to study separability or correlations) is bound to give the same results for $\tr_c(\rho)$ and $\tr_c^\text{Bos.}\mathcal{Q}(\rho)$. As in our article we only use this kind of operations, this implies that our work is devoid of algebraic mistakes.

We have just proven that using $\mathcal{Q}$ guarantees the correctness of our results. As we mentioned before, we can also define other mappings analogous to $\mathcal{Q}$, but which take bases different from \eq{1} as the starting point for their definition. For instance, we could rearrange the operators in \eq{1} in an arbitrary ways, and then by an analogous definition to that in section \ref{qdef} we would obtain different mappings, say $\mathcal{D}_1$ and  $\mathcal{D}_2$ among others.

The first section of our work (more concretely, eqs. (5-6) and their associated discussion), where \cite{br} considers that there are algebraic mistakes, is just a simple proof by counterexample of the fact that, in general,
\begin{equation}\mathcal{D}_1^{-1}\vert_a\circ\tr^\text{Bos.}_c\circ\mathcal{D}_1\neq\mathcal{D}_2^{-1}\vert_a\circ\tr^\text{Bos.}_c\circ\mathcal{D}_2.\end{equation}

 In other words, all \cite{Mig2} does is to consider the class of all these mappings $\mathcal{D}_i$ (each identified by the fermionic operator ordering they come from) and analyses the entanglement properties of the qubit systems they give rise to. This does not lead to physical results in general, but it is what has been done before in the literature. We also provide a particular class of operator orderings (and their associated $\mathcal{D}_i$) which can be regarded as physical; when we perform the partial trace in the qubit state, the states obtained are exactly those obtained with the braided tensor product formalism in \cite{br}. 
 
 A simple proof of the validity of the formalism developed in \cite{Mig2} can be seen in \cite{Mig5}, where it is shown that in general entanglement of fermionic systems must converge in the infinite acceleration limit. This work shows how the measures computed using nonphysical operator orderings do not comply with this requirement, while measures computed using the physical ordering \cite{Mig2} do.
 
Considering these various unphysical orderings is far from being a fruitless mathematical exercise; endowing a fermionic system with some particular tensor product structure may prove quite useful, for instance, when computing the field vacuum and excitations. This was the main reason for doing so, as can be seen in our paper as well as in a number of other works \cite{AlsingSchul,Edu2,Edu3,allicupahc,allicupahc2,Shapoor,Geneferm,chor1,Edu9,chor2}. Even \cite{Cirac} (a seminal article cited in \cite{br}) uses a definite tensor product structure defined on a fermionic Hilbert space to provide a separability criterion for  fermionic systems.

We are not claiming, as \cite{br} suggests, that endowing a tensor product structure a fermionic field becomes a bosonic one. This would indeed lead to a number of paradoxes thoroughly pointed out by \cite{br}. However, this is not what we did in \cite{Mig2}. Rather, we established a correspondence between fermionic and bosonic fields and study how physical results for the fermionic case may be obtained from the corresponding bosonic systems.

To sum up, what \cite{br}  discusses is just another formalism for studying the same issues, not unlike other approaches such as defining entanglement only through observables as it is done in \cite{Juanjo1}, or through observable-induced tensor products \cite{Juanjo2}. Furthermore, the endowing of tensor product structures on fermionic systems as a separability criterion is not new \cite{Cirac}. Our aim when we wrote the paper was to make it accessible for all the people who had written papers about fermionic entanglement in non-inertial frames without taking proper care of the problems regarding the tensor product structure of fermionic systems. 

We realise that, maybe, the choice of title and abstract in \cite{Mig2} was not very fortunate. Instead of calling it `ambiguity in noninertial frames', which  may suggest that we were \emph{reporting} a real ambiguity, we should have emphasised the fact that we were pointing out and correcting an issue present in previous literature on fermionic entanglement in noninertial frames.

\section{Superselection rules and entanglement measures}\label{super}

The author of \cite{br} notes that the state in equation (3) of reference \cite{Mig2}, namely
\begin{align}\ket{\Psi}&=\frac{1}{2}\left(\ket{00}+\ket{01}+\ket{10}+\ket{11}\right)\label{st3}\end{align} with
\begin{align} \ket{00}&=\ket{0},\; \ket{10}=a^\dagger\ket{0},\;\ket{01}=b^\dagger\ket{0},\nonumber\\\ket{11}&=a^\dagger b^\dagger\ket{0}\end{align}
 does not comply with the fermionic superselection rule. This is true, but not relevant for the argument. As reference \cite{weinbergo} explains, the issue of wether or not a certain field state can or cannot be formed in practice cannot be settled by reference to symmetry principles resulting in superselection rules (in the case at hand, these would be rotational and CP symmetries), since we can always enlarge the symmetry group to a new one which lacks these superselection rules \cite{weinbergo}.

In any case, we were not claiming to be able to construct such states. In section of \cite{Mig2} we were merely presenting an academical example and not studying physical field states yet; the state \eqref{st1} reflects the fact that a state (a legitimate element of the Fock space) may be regarded as separable or entangled, if we map it to a qubit system, depending on the operator ordering chosen. In any case, \eqref{st1} or the three-mode state  (eq. (10) of \cite{br} and (5) of \cite{Mig2}) are just examples to illustrate the kind of manipulations which are employed in \cite{Mig2}.  Although the very general family of states presented in \cite{Mig2} does not fulfill the superselection rule, all the particular physical field states studied in that work do comply with it. This is also the case of the state \eqref{st1} which is nothing but a reformulation of the example state \eqref{st3} but compliant with the superselection rule.
 
Continuing with this point, \cite{br}  claims that neither the Grassman field state
\begin{align}\ket{\Psi}=\frac{1}{\sqrt{2}}\left(\ket{0}_\text{A}\ket{0}_\text{R}+\ket{1}_\text{A}\ket{1}_\text{R}\right)\label{jaja1}\end{align}
nor the Dirac state
\begin{align}\ket{\Psi}=\frac{1}{\sqrt{2}}\left(\ket{\up}_\text{A}\ket{\down}_\text{R}+\ket{\down}_\text{A}\ket{\up}_\text{R}\right)\label{jaja2}\end{align}
are analysed in \cite{Mig2}. However, state \eqref{jaja1} corresponds to figure 2 of \cite{Mig2}, whereas state \eqref{jaja2} corresponds to figure 3 of that work. We acknowledge that this last case may be confusing due to the fact that in eq. (20) of \cite{Mig2}, where Alice's basis should have been denoted as $\{\ket{\up},\ket{\down}\}$ instead of $\{\ket{0},\ket{1}\}$.

In spite of these answers, \cite{br} indeed raised an important style flaw in \cite{Mig2}: We introduced states and families of states not compliant with the superselection rule, when this was not necessary at all for our arguments or conclusions. In one case, the use of states violating the rule allowed us to present clearer examples of the phenomena under consideration; in the other, it was done only for the sake of generality, although the particular cases considered were compliant with the rule. But the introduction of these states has led the discussion far from the main points raised in \cite{Mig2}, and the derived confusion outweighs the benefits obtained by far. If we were to write \cite{Mig2} anew, we definitely would not use the same examples.

After this point, the criticism in \cite{br} broadens from \cite{Mig2} to a questioning of the entanglement measures employed in many previous works in field entanglement in non-inertial frames.  \cite{br} claims that, as a consequence of the superselection rule, some entanglement measures are ill-defined. While it is true that these entanglement measures cannot be interpreted exactly in the same way as for bosonic fields, some points of this interpretation remain, however: If a fermionic bipartite field state $\rho$ is separable, meaning that it may be written as
\begin{align}\rho=\sum_i p_i A_i B_i\proj{0}{0}B_iA_i,\quad p_i>0\label{st2}\end{align}
with $A_i$ containing only creation operators of the first bipartition, and $B_i$ containing only operators of the second, then it has vanishing negativity. The converse statement also holds for $2\times2$ and $2\times3$ dimensional systems.

Furthermore, \cite{br} asserts that the superselection rule  invalidates the use of entanglement measures in fermionic fields, save in a few specific cases. However, this criticism is only relevant as long as we consider field states with different superselection parities for each of the parties. If we consider the field state of \cite{Edu9}
\begin{align}\ket{\Psi}=\frac{1}{\sqrt{2}}\left(\ket{0}_\text{A}\ket{0}_\text{R}+\ket{1}_\text{A}\ket{1}_\text{R}\right)\end{align}
then it is true that a Fock space basis for any of the parties (Alice or Rob) must include states in different superselection sectors, such as $\ket{0}$ and $\ket{1}$. But if following \cite{Mig2} we consider e.g. a Dirac singlet state, such as
\begin{align}\ket{\Psi}=\frac{1}{\sqrt{2}}\left(\ket{\up}_\text{A}\ket{\down}_\text{R}+\ket{\down}_\text{A}\ket{\up}_\text{R}\right)\end{align}
then the Fock space basis for any party is composed of states in only one superselection sector. In this case, entanglement measures such as negativity recover exactly the same meaning they would have for qubit systems, since the superselection rule now becomes trivial.

Finally, \cite{br} ellaborates over the various benefits of studying field entanglement in non-inertial frames with the tools of quantum Shannon theory. Though we acknowledge these benefits, we see it as a complementary tool rather than an invalidation of the use of entanglement measures in noninertial frames. Negativity (the measure we use in \cite{Mig2}) for instance is a direct measure of the distillable entanglement present in the state under consideration, which may be purified by means of standard protocols\cite{purif}. 

It can be shown, however, that the qualitative results are not very dependent on the entanglement measure employed either. For instance, \cite{br} computes the entanglement of formation (that has a clear operational meaning) in the same setting we consider, and there are no qualitative differences with our result: the behaviour of the entanglement of formation is qualitatively the same as negativity. 

Independently of the many uses quantum Shannon theory may have, entanglement measures provide a clear and straightforward way to study quantum correlations between two parties; this, and nothing else, is what has been done in the numerous works on the subject.

\section{Conclusions}\label{cojoniak}

We have shown that there are no algebraic mistakes in \cite{Mig2}. We are not directly studying a fermionic system, but rather a qubit one obtained by means of a well-defined mapping. We simply characterise the set of these mappings and find the circumstances under which they give rise to physical results. The mere existence of \cite{br}, however, proves that we have failed in our attempt to transmit our ideas in a clear and concise way. It is our hope that this comment may shed some light on the more confusing parts of \cite{Mig2}.

It is true that we use states violating the superselection rule, but do so only for an academical example in which the parity of the fermionic states plays no role. The main results of \cite{Mig2} all comply with the superselection rule. Nevertheless, all of our results could (and should) have been presented with states compliant with the superselection rule, thus avoiding these concerns. 

Also, although the interpretation of entanglement measures for fermionic systems may not be the same as for qubits in general, there are field states for which such an interpretation may be recovered (such a Dirac singlet state).

Finally, though quantum Shannon theory may be an useful tool to study quantum information transmission, entanglement measures provide a clear and straightforward way to study quantum correlations between two parties, wether in a relativistic setting or not.

\section{Acknowledgements}
We thank Juan Le\'on for our interesting and useful discussions related to this work.  


%\bibliography{references}

\begin{thebibliography}{14}%
\makeatletter
\providecommand \@ifxundefined [1]{%
 \@ifx{#1\undefined}
}%
\providecommand \@ifnum [1]{%
 \ifnum #1\expandafter \@firstoftwo
 \else \expandafter \@secondoftwo
 \fi
}%
\providecommand \@ifx [1]{%
 \ifx #1\expandafter \@firstoftwo
 \else \expandafter \@secondoftwo
 \fi
}%
\providecommand \natexlab [1]{#1}%
\providecommand \enquote  [1]{``#1''}%
\providecommand \bibnamefont  [1]{#1}%
\providecommand \bibfnamefont [1]{#1}%
\providecommand \citenamefont [1]{#1}%
\providecommand \href@noop [0]{\@secondoftwo}%
\providecommand \href [0]{\begingroup \@sanitize@url \@href}%
\providecommand \@href[1]{\@@startlink{#1}\@@href}%
\providecommand \@@href[1]{\endgroup#1\@@endlink}%
\providecommand \@sanitize@url [0]{\catcode `\\12\catcode `\$12\catcode
  `\&12\catcode `\#12\catcode `\^12\catcode `\_12\catcode `\%12\relax}%
\providecommand \@@startlink[1]{}%
\providecommand \@@endlink[0]{}%
\providecommand \url  [0]{\begingroup\@sanitize@url \@url }%
\providecommand \@url [1]{\endgroup\@href {#1}{\urlprefix }}%
\providecommand \urlprefix  [0]{URL }%
\providecommand \Eprint [0]{\href }%
\providecommand \doibase [0]{http://dx.doi.org/}%
\providecommand \selectlanguage [0]{\@gobble}%
\providecommand \bibinfo  [0]{\@secondoftwo}%
\providecommand \bibfield  [0]{\@secondoftwo}%
\providecommand \translation [1]{[#1]}%
\providecommand \BibitemOpen [0]{}%
\providecommand \bibitemStop [0]{}%
\providecommand \bibitemNoStop [0]{.\EOS\space}%
\providecommand \EOS [0]{\spacefactor3000\relax}%
\providecommand \BibitemShut  [1]{\csname bibitem#1\endcsname}%
\let\auto@bib@innerbib\@empty
%</preamble>
\bibitem [{\citenamefont {Montero}\ and\ \citenamefont
  {Mart\'\i{}n-Mart\'\i{}nez}(2011)}]{Mig2}%
  \BibitemOpen
  \bibfield  {author} {\bibinfo {author} {\bibfnamefont {M.}~\bibnamefont
  {Montero}}\ and\ \bibinfo {author} {\bibfnamefont {E.}~\bibnamefont
  {Mart\'\i{}n-Mart\'\i{}nez}},\ }\href@noop {} {\bibfield  {journal} {\bibinfo
   {journal} {Phys. Rev. A}\ }\textbf {\bibinfo {volume} {83}},\ \bibinfo
  {pages} {062323} (\bibinfo {year} {2011})}\BibitemShut {NoStop}%
\bibitem [{\citenamefont {Br\'adler}(2011)}]{br}%
  \BibitemOpen
  \bibfield  {author} {\bibinfo {author} {\bibfnamefont {K.}~\bibnamefont
  {Br\'adler}},\ }\href@noop {} {\  (\bibinfo {year} {2011})},\ \Eprint
  {http://arxiv.org/abs/1108.5553} {arXiv:1108.5553} \BibitemShut {NoStop}%
  \bibitem [{\citenamefont {Br\'adler}\ and\ \citenamefont
  {J\'auregui}(2011)}]{br2}%
  \BibitemOpen
  \bibfield  {author} {\bibinfo {author} {\bibfnamefont {K.}~\bibnamefont
  {Br\'adler}}\ and\ \bibinfo {author} {\bibfnamefont {R.}~\bibnamefont
  {J\'auregui}},\ }\href@noop {} {\bibfield  {journal} {\bibinfo
   {journal} {Phys. Rev. A }\ }\textbf {\bibinfo {volume} {85}}\ \bibinfo
  {pages} {016301} (\bibinfo {year} {2012}) arXiv:1201.1045}\BibitemShut {NoStop}%
\bibitem [{\citenamefont {Alsing}\ \emph {et~al.}(2006)\citenamefont {Alsing},
  \citenamefont {Fuentes-Schuller}, \citenamefont {Mann},\ and\ \citenamefont
  {Tessier}}]{AlsingSchul}%
  \BibitemOpen
  \bibfield  {author} {\bibinfo {author} {\bibfnamefont {P.~M.}\ \bibnamefont
  {Alsing}}, \bibinfo {author} {\bibfnamefont {I.}~\bibnamefont
  {Fuentes-Schuller}}, \bibinfo {author} {\bibfnamefont {R.~B.}\ \bibnamefont
  {Mann}}, \ and\ \bibinfo {author} {\bibfnamefont {T.~E.}\ \bibnamefont
  {Tessier}},\ }\href@noop {} {\bibfield  {journal} {\bibinfo  {journal} {Phys.
  Rev. A}\ }\textbf {\bibinfo {volume} {74}},\ \bibinfo {pages} {032326}
  (\bibinfo {year} {2006})}\BibitemShut {NoStop}%
\bibitem [{\citenamefont {Le\'on}\ and\ \citenamefont
  {Mart\'\i{}n-Mart\'\i{}nez}(2009)}]{Edu2}%
  \BibitemOpen
  \bibfield  {author} {\bibinfo {author} {\bibfnamefont {J.}~\bibnamefont
  {Le\'on}}\ and\ \bibinfo {author} {\bibfnamefont {E.}~\bibnamefont
  {Mart\'\i{}n-Mart\'\i{}nez}},\ }\href@noop {} {\bibfield  {journal} {\bibinfo
   {journal} {Phys. Rev. A}\ }\textbf {\bibinfo {volume} {80}},\ \bibinfo
  {pages} {012314} (\bibinfo {year} {2009})}\BibitemShut {NoStop}%
\bibitem [{\citenamefont {Mart\'\i{}n-Mart\'\i{}nez}\ and\ \citenamefont
  {Le\'on}(2009)}]{Edu3}%
  \BibitemOpen
  \bibfield  {author} {\bibinfo {author} {\bibfnamefont {E.}~\bibnamefont
  {Mart\'\i{}n-Mart\'\i{}nez}}\ and\ \bibinfo {author} {\bibfnamefont
  {J.}~\bibnamefont {Le\'on}},\ }\href@noop {} {\bibfield  {journal} {\bibinfo
  {journal} {Phys. Rev. A}\ }\textbf {\bibinfo {volume} {80}},\ \bibinfo
  {pages} {042318} (\bibinfo {year} {2009})}\BibitemShut {NoStop}%
\bibitem [{\citenamefont {Pan}\ and\ \citenamefont
  {Jing}(2008{\natexlab{a}})}]{allicupahc}%
  \BibitemOpen
  \bibfield  {author} {\bibinfo {author} {\bibfnamefont {Q.}~\bibnamefont
  {Pan}}\ and\ \bibinfo {author} {\bibfnamefont {J.}~\bibnamefont {Jing}},\
  }\href@noop {} {\bibfield  {journal} {\bibinfo  {journal} {Phys. Rev. A}\
  }\textbf {\bibinfo {volume} {77}},\ \bibinfo {pages} {024302} (\bibinfo
  {year} {2008}{\natexlab{a}})}\BibitemShut {NoStop}%
\bibitem [{\citenamefont {Pan}\ and\ \citenamefont
  {Jing}(2008{\natexlab{b}})}]{allicupahc2}%
  \BibitemOpen
  \bibfield  {author} {\bibinfo {author} {\bibfnamefont {Q.}~\bibnamefont
  {Pan}}\ and\ \bibinfo {author} {\bibfnamefont {J.}~\bibnamefont {Jing}},\
  }\href@noop {} {\bibfield  {journal} {\bibinfo  {journal} {Phys. Rev. D}\
  }\textbf {\bibinfo {volume} {78}},\ \bibinfo {pages} {065015} (\bibinfo
  {year} {2008}{\natexlab{b}}).}
  
  \bibitem[{\citenamefont{Moradi}(2009)}]{Shapoor}
\bibinfo{author}{\bibfnamefont{S.}~\bibnamefont{Moradi}},
  \bibinfo{journal}{Phys. Rev. A} \textbf{\bibinfo{volume}{79}},
  \bibinfo{pages}{064301} (\bibinfo{year}{2009}).
  
\bibitem [{\citenamefont {Ostapchuk}\ and\ \citenamefont
  {Mann}(2009)}]{Geneferm}%
  \BibitemOpen
  \bibfield  {author} {\bibinfo {author} {\bibfnamefont {D.~C.~M.}\
  \bibnamefont {Ostapchuk}}\ and\ \bibinfo {author} {\bibfnamefont {R.~B.}\
  \bibnamefont {Mann}},\ }\href@noop {} {\bibfield  {journal} {\bibinfo
  {journal} {Phys. Rev. A}\ }\textbf {\bibinfo {volume} {79}},\ \bibinfo
  {pages} {042333} (\bibinfo {year} {2009})}\BibitemShut {NoStop}%
\bibitem [{\citenamefont {Wang}\ and\ \citenamefont {Jing}(2010)}]{chor1}%
  \BibitemOpen
  \bibfield  {author} {\bibinfo {author} {\bibfnamefont {J.}~\bibnamefont
  {Wang}}\ and\ \bibinfo {author} {\bibfnamefont {J.}~\bibnamefont {Jing}},\
  }\href@noop {} {\bibfield  {journal} {\bibinfo  {journal} {Phys. Rev. A}\
  }\textbf {\bibinfo {volume} {82}},\ \bibinfo {pages} {032324} (\bibinfo
  {year} {2010})}\BibitemShut {NoStop}%
\bibitem [{\citenamefont {Bruschi}\ \emph {et~al.}(2010)\citenamefont
  {Bruschi}, \citenamefont {Louko}, \citenamefont {Mart\'\i{}n-Mart\'\i{}nez},
  \citenamefont {Dragan},\ and\ \citenamefont {Fuentes}}]{Edu9}%
  \BibitemOpen
  \bibfield  {author} {\bibinfo {author} {\bibfnamefont {D.~E.}\ \bibnamefont
  {Bruschi}}, \bibinfo {author} {\bibfnamefont {J.}~\bibnamefont {Louko}},
  \bibinfo {author} {\bibfnamefont {E.}~\bibnamefont
  {Mart\'\i{}n-Mart\'\i{}nez}}, \bibinfo {author} {\bibfnamefont
  {A.}~\bibnamefont {Dragan}}, \ and\ \bibinfo {author} {\bibfnamefont
  {I.}~\bibnamefont {Fuentes}},\ }\href@noop {} {\bibfield  {journal} {\bibinfo
   {journal} {Phys. Rev. A}\ }\textbf {\bibinfo {volume} {82}},\ \bibinfo
  {pages} {042332} (\bibinfo {year} {2010})}\BibitemShut {NoStop}%
\bibitem [{\citenamefont {Khan}\ and\ \citenamefont {Khan}(2011)}]{chor2}%
  \BibitemOpen
  \bibfield  {author} {\bibinfo {author} {\bibfnamefont {S.}~\bibnamefont
  {Khan}}\ and\ \bibinfo {author} {\bibfnamefont {M.~K.}\ \bibnamefont
  {Khan}},\ }\href@noop {} {\bibfield  {journal} {\bibinfo  {journal} {J. of
  Phys. A}\ }\textbf {\bibinfo {volume} {44}} (\bibinfo {year}
  {2011})}\BibitemShut {NoStop}%
  \bibitem [{\citenamefont {Ba\~nuls}\ \emph {et~al.}(2007)\citenamefont
  {Ba\~nuls}, \citenamefont {Cirac},\ and\ \citenamefont {Wolf}}]{Cirac}%
  \BibitemOpen
  \bibfield  {author} {\bibinfo {author} {\bibfnamefont {M.-C.}\ \bibnamefont
  {Ba\~nuls}}, \bibinfo {author} {\bibfnamefont {J.~I.}\ \bibnamefont {Cirac}},
  \ and\ \bibinfo {author} {\bibfnamefont {M.~M.}\ \bibnamefont {Wolf}},\
  }\href@noop {} {\bibfield  {journal} {\bibinfo  {journal} {Phys. Rev. A},\
  }\textbf {\bibinfo {volume} {76}},\ \bibinfo {pages} {022311} (\bibinfo
  {year} {2007})}\BibitemShut {NoStop}%
 \bibitem [{\citenamefont {Montero}\ and\ \citenamefont
  {Mart\'\i{}n-Mart\'\i{}nez}(2011)}]{Mig5}%
  \BibitemOpen
  \bibfield  {author} {\bibinfo {author} {\bibfnamefont {M.}~\bibnamefont
  {Montero}}\ and\ \bibinfo {author} {\bibfnamefont {E.}~\bibnamefont
  {Mart\'\i{}n-Mart\'\i{}nez}},\ }\href@noop {} {\bibfield  {journal} {\bibinfo
   {journal} {arxiv:1111.6070}\ } (\bibinfo {year} {2011})}\BibitemShut {NoStop}%
\bibitem [{\citenamefont {Barnum}\ \emph {et~al.}(2004)\citenamefont {Barnum},
  \citenamefont {Knill}, \citenamefont {Ortiz}, \citenamefont {Somma},\ and\
  \citenamefont {Viola}}]{Juanjo1}%
  \BibitemOpen
  \bibfield  {author} {\bibinfo {author} {\bibfnamefont {H.}~\bibnamefont
  {Barnum}}, \bibinfo {author} {\bibfnamefont {E.}~\bibnamefont {Knill}},
  \bibinfo {author} {\bibfnamefont {G.}~\bibnamefont {Ortiz}}, \bibinfo
  {author} {\bibfnamefont {R.}~\bibnamefont {Somma}}, \ and\ \bibinfo {author}
  {\bibfnamefont {L.}~\bibnamefont {Viola}},\ }\href@noop {} {\bibfield
  {journal} {\bibinfo  {journal} {Phys. Rev. Lett.},\ }\textbf {\bibinfo
  {volume} {92}},\ \bibinfo {pages} {107902} (\bibinfo {year}
  {2004})}\BibitemShut {NoStop}%
\bibitem [{\citenamefont {Zanardi}\ \emph {et~al.}(2004)\citenamefont
  {Zanardi}, \citenamefont {Lidar},\ and\ \citenamefont {Lloyd}}]{Juanjo2}%
  \BibitemOpen
  \bibfield  {author} {\bibinfo {author} {\bibfnamefont {P.}~\bibnamefont
  {Zanardi}}, \bibinfo {author} {\bibfnamefont {D.~A.}\ \bibnamefont {Lidar}},
  \ and\ \bibinfo {author} {\bibfnamefont {S.}~\bibnamefont {Lloyd}},\
  }\href@noop {} {\bibfield  {journal} {\bibinfo  {journal} {Phys. Rev.
  Lett.},\ }\textbf {\bibinfo {volume} {92}},\ \bibinfo {pages} {060402}
  (\bibinfo {year} {2004})}\BibitemShut {NoStop}%
  \bibitem [{\citenamefont {Weinberg}(1995)\citenamefont
  {Weinberg}}]{weinbergo}%
  \BibitemOpen
  \bibfield  {author} {\bibinfo {author} {\bibfnamefont {S.}~\bibnamefont
  {Weinberg}}},\
  \href@noop {} {\bibfield  {journal} {\bibinfo  {journal} {The Quantum Theory of Fields I}\ }
  (\bibinfo {year} {Cambridge University Press, 1995})}\BibitemShut {NoStop}%
\bibitem [{\citenamefont {Pan}\ \emph {et~al.}(2006)\citenamefont {Pan},
  \citenamefont {Simon}, \citenamefont {Brukner},\ and\ \citenamefont
  {Zeilinger}}]{purif}%
  \BibitemOpen
  \bibfield  {author} {\bibinfo {author} {\bibfnamefont {J.-~W.}\ \bibnamefont
  {Pan}}, \bibinfo {author} {\bibfnamefont {C.}~\bibnamefont
  {Simon}}, \bibinfo {author} {\bibfnamefont {\u{C}.}\ \bibnamefont
  {Brukner}}, \ and\ \bibinfo {author} {\bibfnamefont {A.}\ \bibnamefont
  {Zeilinger}},\ }\href@noop {} {\bibfield  {journal} {\bibinfo  {journal} {Nature}
  \ }\textbf {\bibinfo {volume} {410}},\ \bibinfo {pages} {1067-1070}
  (\bibinfo {year} {2001})}\BibitemShut {NoStop}%
%\bibitem [{\citenamefont {Br\'adler}(2011)}]{bradxiv}%
 % \BibitemOpen
 % \bibfield  {author} {\bibinfo {author} {\bibfnamefont {K.}~\bibnamefont
 % {Br\'adler}},\ }\href@noop {} {\  (\bibinfo {year} {2011})},\ \Eprint
 % {http://arxiv.org/abs/1108.5553} {arXiv:1108.5553} \BibitemShut {NoStop}%
% \bibitem [{\citenamefont {Montero}\ and\ \citenamefont
 % {Mart\'\i{}n-Mart\'\i{}nez}(2011)}]{pc}%
  %\BibitemOpen
  %\bibfield  {author} {\bibinfo {author} {\bibfnamefont {M.}~\bibnamefont
  %{Montero}}\ and\ \bibinfo {author} {\bibfnamefont {E.}~\bibnamefont
  %{Mart\'\i{}n-Mart\'\i{}nez}},\ }\href@noop {} {}
  %{journal report} \BibitemShut {NoStop}%
\end{thebibliography}

\end{document}